\documentclass[conference]{IEEEtran}

\makeatletter
\def\ps@headings{%
\def\@oddhead{\mbox{}\scriptsize\rightmark \hfil \thepage}%
\def\@evenhead{\scriptsize\thepage \hfil \leftmark\mbox{}}%
\def\@oddfoot{}%
\def\@evenfoot{}}
\makeatother

\usepackage{cite,pinlabel,url,graphicx}
\usepackage[cmex10]{amsmath}
\newtheorem{theorem}{Theorem}
\newtheorem{defn}{Definition}
\newtheorem{example}{Example}

\newcommand{\Esym}{\mathrm{E}}
\newcommand{\E}[1]{\Esym\left[#1\right]}

\begin{document}
%
\title{Measuring the likelihood of models for network evolution}

\author{
\IEEEauthorblockN{Richard G. Clegg, Raul Landa}
\IEEEauthorblockA{Dept of Elec. Eng.\\
University College London\\
London, UK\\
Email: richard@richardclegg.org,\\
Email: rlanda@ee.ucl.ac.uk}
\and
\IEEEauthorblockN{Hamed Haddadi}
\IEEEauthorblockA{
Max Planck Institute  for \\
Software Systems (MPI-SWS) \\
Kaiserslautern/Saarbr\"ucken Germany\\
Email : hamed@mpi-sws.org
}
\and
\IEEEauthorblockN{Miguel Rio}
\IEEEauthorblockA{Dept of Elec. Eng.\\
University College London\\
London, UK\\
Email: m.rio@ee.ucl.ac.uk}
}

\maketitle

\begin{abstract}
Many researchers have hypothesised models which explain the evolution of
the topology of a target network.  The framework described in this paper
gives the likelihood that the target network arose from the hypothesised
model.  This allows rival hypothesised models to be compared for their
ability to explain the target network.  A null model (of random evolution)
is proposed as a baseline for comparison.  The framework also considers models
made from linear combinations of model components.  A method is given for
the automatic optimisation of component weights.  The framework is tested
on simulated networks with known parameters and also on real data.
\end{abstract}

\section{Introduction}
\label{sec:intro}
The field of modelling graph topologies (and in particular the topology 
of the Internet) has generated a huge degree of research interest 
in recent years (see \cite[chapter 3]{handbook} for a review of the subject
and \cite{hamedsurvey}
for an Internet topology perspective).  This paper introduces 
FETA (Framework for Evolving Topology Analysis) which
can be used to assess
potential underlying models for any network where information about the network
evolution is available.  Previously, many researchers have fitted 
probabilistic topology models by growing candidate models and 
assessing how well their model fitted against a selection of statistics 
made on a snapshot of the real network.  The
FETA approach, by contrast, uses a single statistic to get a 
rigorous estimate for the likelihood of a model based upon the dynamic 
evolution of the network.  This paper concentrates on results from artificial
models proving the framework reproduces known models.  A companion paper
\cite{feta} reports on results from five real networks but does not 
present the artificial test data given here.

It has been known for some time that a number of networks follow an
approximate
power law in their degree distribution.  Such networks
include the Internet 
Autonomous System (AS) topology, world wide web,
co-authorship networks, sexual contact networks, email, networks of actors, networks from biology and many others 
(many references are in \cite[table 3.1]{handbook}).  Researchers
have attempted to grow artificial versions of such networks with
models which assign connection probabilities to existing nodes based upon
the graph topology.  Often surprisingly simple models replicate many
features of real networks, such as power laws. The celebrated 
Barab\'asi--Albert (BA) model \cite{ba} provides
an explanation for these in terms of a ``preferential
attachment" model (the probability of connecting to a node is exactly
proportional to its degree).

Further models have given slightly different probabilities and
slightly different ways of connecting nodes to better match
the statistics of real graphs
\cite{ba,integrated,ba2,bu,zhou2004}.  These models are usually 
assessed
by growing artificial networks and measuring several
representative statistics to compare with the real target network.  
A few models work differently, for example ORBIS \cite{ORBIS}
does not ``grow" a network by link addition but instead ``rescales" it.
Willinger et al \cite{Willinger2002} called for a ``closing of the loop"
with a verification stage which checks how well the proposed model
fits the target network.  FETA addresses this validation problem.  
The FETA procedure evaluates the dynamic
evolution of a network, not a static snapshot.  It directly estimates 
a rigorous likelihood rather than attempting to find several summary 
statistics and this likelihood is estimated directly from the network itself
rather than by growing and measuring an artificial network using the
model to be tested.

\section{Evaluation and optimisation framework}
\label{sec:framework}
Let $G$ be some graph which evolves in time.  Let $G_t$ be the state
of this graph at some step of evolution, $t$.
Consider a model for network evolution as consisting of two separate 
(but interconnected) models.  The {\em outer model\/} 
selects the operation which transforms the graph between two steps.
The {\em inner model\/} chooses the entity for
that operation.  The operation and the entity together define
the transition from $G_{i-1}$ to $G_i$.
Both the outer and inner models may depend on the state of
the graph $G_i$ on the step of the evolution $i$ and possibly on
exogenous parameters.
Outer model operations might be the following:
\begin{enumerate}
\item Add a new node and connect it to an existing node.
\item Connect the newest node to an existing node.
\item Connect two existing nodes.
\item Delete an existing connection.
\item Delete an existing node and its connections.
\end{enumerate}
These outer models work with inner models which select either nodes or
edges for the operation.  The inner model assigns probabilities to each node
(operations 1, 2 and 5)
or edge (operations 3 and 4)\footnote{Note 
that the reason ``add a new node" is
not considered on its own is to confine the study here   
to connected graphs.}.
There may be a different inner model for each
outer model operation.  The outer model might be adapted further
if the known graph
data can include unconnected (degree zero) nodes, if graphs can
be unconnected and so on.  The focus of FETA is the inner model
and the outer model is not discussed here.  

\begin{example}
The BA model \cite{ba} 
has a simple outer model which performs 
step 1) then step 2) twice (a new node connects to exactly
three existing nodes). The inner model, known as preferential
attachment, assigns a probability to each node exactly proportional to its
degree.  This inner model is referred to in this paper as $\theta_d$.
The positive feedback preference (PFP) model 
\cite{zhou2004}, uses a parameterised outer model involving several
connections and an inner model
which assigns node probabilities where the probability of selecting a node
with degree $d$ is proportional to
$d^{1 + \delta \log_{10} (d)}$ where $\delta$ is a parameter. 
\end{example}

\subsection{Evaluating inner model likelihood}
\label{sec:eval}

Let $G_0$ be the graph at the first step of evolution observed
(this need not be right at the start of the evolution of the graph).
Assume that the state of the graph is observed until some step $G_t$.
The graph evolves between step $G_{i-1}$ and $G_i$ according to an outer
and inner model.  Each step involves the addition of one edge.
For simplicity of explanation consider the outer model
to consist only of the two operations:
\begin{enumerate}
\item add a new node and connect it to an existing node $N_i$; or
\item connect the newest node to an existing node $N_i$.
\end{enumerate}
The inner model $\theta$ assigns probabilities to the
existing nodes at a given step.  Given the above outer model, 
from $G_{i-1}$ and $G_i$ the node $N_i$ chosen by the inner 
model can be inferred.  Call the set of all observed choices
$C = (N_1, \ldots,N_t)$.

\begin{defn}
An inner model $\theta$ is a map which at every choice stage
$j$ maps a node $i$ 
to a probability $p_j(i|\theta)$.  A model $\theta$ is a 
{\em valid model\/} if the sum over all nodes is one
$\sum_i p_j(i|\theta)=1$.
\end{defn}

\begin{theorem}\label{thm}
Let $C=(N_1,\ldots,N_t)$ be the observed node choices at steps $1,\ldots,t$ of
the evolution of the graph $G$.  
Let $\theta$ be some hypothesised inner model which assigns
a probability $p_j(i|\theta)$ to node $i$ at step $j$.  The likelihood of
the observed $C$ given $\theta$ is
$$
L(C|\theta) = \prod_{j=1}^t p_j(N_j|\theta).
$$
\end{theorem}

\begin{proof}
If $L(C_j|\theta)$ is the likelihood of the $j$th choice given model $\theta$
then $L(C|\theta) = \prod_{j=1}^t L(C_j|\theta)$.
Given $p_j(N_j|\theta)$ is the probability
model $\theta$ assigns to node $N_j$ at step $j$, therefore it is also
the likelihood of choice $N_j$ at step $j$ given model $\theta$.  The
theorem follows.
\end{proof}

If two inner models $\theta$ and $\theta'$ are hypothesised to
explain the node choices $C$ arising from observations of a graph
$G_0,\ldots,G_t$ and a given outer model, then the one with the
higher likelihood is to be 
preferred\footnote{A model with fewer parameters will
sometimes be preferred if the gain in likelihood is small
or the number of parameters added is large\cite{AIC} --
the extreme case of this is the saturated model $\theta_s$.}.  
In practice, for even moderate
sized graphs, this likelihood will be beyond the computational 
accuracy of most programming languages and the log likelihood
$l(C|\theta) = \log(L(C|\theta))$ is more useful.  

A common statistical measure is the deviance 
$D = -2 l(C|\theta)$.  (The deviance is usually defined with respect
to a ``saturated model" -- in this case the saturated model 
$\theta_s$ is the model which has $p_j(C_j|\theta_s) = 1$ for
all $j \in 1,\ldots,t$ and hence has $l(C|\theta_s) = 0$.  The saturated model
$\theta_s$ has likelihood one but is useless for anything except
exactly reproducing $G_0, \ldots, G_t$).

\begin{defn}
Let $\theta_0$ be the {\em null model}.  Here, an appropriate
null model is the one which assigns equal probability
to all nodes in the choice set (the random model).  
The choice set is either the set
of  all nodes or, if a simple graph is desired, the set of all
nodes to which the new node does not already connect. 
\end{defn}

The null model allows the assessment of the null deviance 
$D_0 = -2 (l(C|\theta) - l (C| \theta_0))$.  However, both
$D$ and $D_0$ depend heavily on the size of $t$ (the number
of choices made).  A more useful measure created for this
situation is now given.

\begin{defn}\label{defn:pclr}
Let $\theta$ be some inner model hypothesis for the set
of node choices $C = (N_1, \ldots, N_t)$.  
Let $\theta_A$ be some rival model to
compare $\theta$ with.  The {\em per choice likelihood ratio\/}
with $\theta_A$, $c_A$, is the likelihood ratio normalised
by $t$ the number of choices.  It is given by
$$
c_A = \left [\frac{L(C|\theta)}{L(C|\theta_A)} \right]^{1/t}
= \exp \left[ \frac{l(C|\theta) - l(C|\theta_A)}{t} \right]. 
$$
\end{defn}

A value $c_A > 1$ indicates that $\theta$ is a better explanatory
model for the choice set $C$ than $\theta_A$ and
$c_A < 1$ indicates it is worse.
Particularly useful is $c_0$ the {\em per choice likelihood ratio
relative to the null model\/}.  
Note that for a fixed $C$, given the $c_0$ statistic 
for two models  $\theta$ and $\theta_A$ 
then $c_A$ can be shown to be the
ratio of the former over the latter.

In summary, the likelihood $L(C|\theta)$ gives the absolute
likelihood of a given model $\theta$ producing the choice set
$C$ arising from a set of graphs $G_0, \ldots, G_t$.  However,
the per choice likelihood ratio produces a result on a more
comprehensible scale.

\subsection{Fitting linear combinations of model components}
\label{sec:linear}

An inner model $\theta$ can be constructed from
a linear combination of other inner models.  
Let $\theta_1$, $\theta_2$,
$\ldots$ be probability models.  A combined model can
now be constructed from component models as follows, $\theta = 
\beta_1 \theta_1 + \beta_2 \theta_2 + \cdots + \beta_N \theta_N$.
The $\beta_i$ are known as the component weights. The model
$\theta$
is a valid model if all $\beta \in (0,1)$ and
$\sum_i \beta_i = 1$. The weights $\beta$ 
that best explain $C$ can be obtained 
using a fitting
procedure from statistics known as Generalised Linear Models
(GLM).

Let 
$
P_j(i) = 
1$ if  $i = N_j,$ and $P_j(i) = 0$ otherwise.
The problem of finding the best model weights
becomes the problem of fitting the GLM,
$
P_j(i) = \beta_1 p_j(i|\theta_1) + \beta_2 p_j(i|\theta_2) + 
\cdots + \varepsilon.
$
A GLM procedure can fit the $\beta$  parameters to
find the combined model $\theta$ which best fits the $P_j(i)$.
This fit is obtained by creating a data point 
for each choice $j$ and
for each node $i$ 
giving information about that node at that
choice time and also the value of $P_j(i)$. 

GLM fitting in a statistical language such as 
R\footnote{\url{http://www.r-project.org/}} 
can be used to find the choice of 
$\beta_i$ which maximises the likelihood of this model.  
This
is equivalent to finding the $\beta_i$ which gives the
maximum likelihood for $\theta$ since for model $\theta$,
the expectation $\E{P_j(i)} = p_j(i|\theta)$.
The fitting procedure estimates for each $\beta_i$, the
value, the error and the statistical significance.

Because this procedure requires one line of data for each node
at each choice then it produces a large amount of data and sampling
is necessary.  As will be seen in section \ref{sec:GLMtest} the method
still recovers parameters accurately.

\subsection{FETA in practice}
\label{sec:feta_general}

For simplicity of
discussion in previous sections, only operations which connected a new
node to a single node were considered.  Using the framework to connect
edges between existing internal nodes requires a small extension.
Since the number of potential edges is roughly the square of the number of 
nodes, it makes sense to decompose the choice of an edge into the
choice of a start node and an end node.  Once a start node is picked,
the choice set for the end node can be constrained to ensure the graph 
remains simple.  The likelihood of adding edge
$(x,y)$ is calculated as the likelihood of choosing node $x$ then node $y$
plus the likelihood of choosing node $y$ then node $x$.  For the
purposes of definition \ref{defn:pclr} an edge counts as two choices (since
definition \ref{defn:pclr} is in terms of node choices).

The outer model could be further generalised by, for example, adding the
possibility of a ``bare" node appearing (a node with no links) if this
event could be observed.  Another extension would be 
adding node or edge deletion operations.
Separate inner models can be fitted to different outer model operations.
For example, in the work on FETA reported in \cite{feta} separate models
are fitted to the outer model operations which connect a single existing 
node to a new node and the outer model operations which connect an
edge between existing internal nodes.  Likelihoods from the two parts
of the inner model can be directly combined by multiplication.

Another practical concern is scalability -- how the likelihood 
computation time increases as graphs become large. 
Tests were run on a 2.66GHz quad core Xeon CPU using the
same codebase for two tasks, one to measure the likelihood of
a target network arising from a given model and the second
to actually create a network.  The number of links created
was varied from 1,000 to 100,000.  While both processes
increased approximately as $O(n^2)$ where $n$ is the number
of links,
the likelihood calculation
is much quicker than the network creation process.
For 100,000 links the likelihood calculation took 53 seconds, the
network creation took 2,600 seconds.  Compared with producing a
test network and measuring it, the FETA approach is extremely efficient.
If the runtime were to become onerous, sampling could be used
as it is in the GLM procedure.  This was not necessary for the results
in this paper.

It is worth briefly noting two points about data requirements.
Firstly, FETA does not require data from the entire history of a network,  
the graph $G_0$ can be any stage of graph construction.  Secondly, for a
sufficiently large graph, knowing the exact order of link arrival 
should not be necessary (this may occur if the graph state is measured
periodically rather than recorded as every node or edge arrives). 
A graph with a large number of nodes will 
not change its topology greatly for a small number of arrivals and therefore
a small reordering of link arrival order should make little difference 
to the model likelihood.  Future work will seek to quantify the inaccuracies  
introduced by this reordering.

\section{Testing the framework}
\label{sec:testing}
The obvious way to test the framework is on simulated data sets
where the underlying inner model is known.  Testing models using
the likelihood procedure from \ref{sec:eval}
is demonstrated in section \ref{sec:likelihoodtest}.  
Optimising models using the GLM procedure in section 
\ref{sec:linear} is done in section \ref{sec:GLMtest}.
A demonstration on real data is described in section 
\ref{sec:realdata}.

Let $d_i$ be the degree of node $i$ and $t_i$ be the triangle
count (the number of triangles, or 3--cycles, the node
is in).
The model components used in the testing are the following:
$\theta_0$ -- the null model (random model) assumes all nodes have equal 
probability $p_i = k_n$;
$\theta_d$ -- the degree model (preferential attachment) 
assumes node probability $p_i = k_d d_i$;
$\theta_t$ -- the triangle model assumes node 
probability $p_i = k_t t_i$;
$\theta_S$ -- the singleton model assumes node 
probability $p_i = k_S$ if $d_i = 1$ and $p_i = 0$ otherwise;
$\theta_D$ -- the doubleton model assumes node 
probability $p_i = k_D$ if $d_i = 2$ and $p_i = 0$ otherwise;
$\theta_R(n)$ -- the ``recent" model where $p_i = k_H$ if a node
was one selected in the last $n$ selections and $p_i = 0$ otherwise
and
$\theta_p^{(\delta)}$ -- the 
PFP model assumes node 
probability $p_i = k_p d_i^{1 + \delta \log_{10}(d_i)}$.
The $k_{\bullet}$ 
are all normalising constants to ensure $\sum_i p_i = 1$.

\subsection{Testing the likelihood framework}
\label{sec:likelihoodtest}

The best way to test the likelihood framework is on simulated
networks with a known underlying inner model.  Test model one
has a simple outer model which creates
a new node and then connects it to exactly three
distinct nodes.  
The inner model $\theta_1$ which
chooses these nodes is $\theta_1 = 0.5 \theta_p(0.05) + 0.5 \theta_t$.
That is, it is 50\% the PFP inner model with $\delta = 0.05$
and 50\% the triangle model.  Naturally, nodes with a high number 
of triangles also have a high degree so these model parameters
are, to some extent, correlated.

An artificial network was grown with 10,000 edges using the
model described above.  Assuming that the model was known
to be of the form $\beta_p \theta_p(\delta) + \beta_t \theta_t$
then, since $\beta_p + \beta_t = 1$ a sweep of the parameters
$\delta$ and $\beta_t$ should give a likelihood surface 
with a maximum at the correct values of $\beta_t$ and $\delta$.
The values tried were all possible combinations of
$\beta_t = (0.1, 0.15, \ldots, 0.85,0.9)$ and
$\delta = (0.01, 0.0125, \ldots, 0.0875,0.09)$.
The likelihood surface produced
is shown in Figure \ref{fig:likelihood}
with contour lines projected below.  As can be seen, the maximum
likelihood is in the correct part of the region
($\beta_t = 0.5$, $\delta= 0.05$).
In fact the highest $c_0$ was with $\delta = 0.0525$ and
$\beta_t = 0.5$, an almost exact recovery of the correct
parameters.

\begin{figure}[ht!]
\begin{center}
\labellist
\pinlabel {$\delta$} at 90 130
\pinlabel {$\beta_t$} at 350 95
\pinlabel {$c_0$} at 100 260
\endlabellist
\includegraphics[width=8cm]{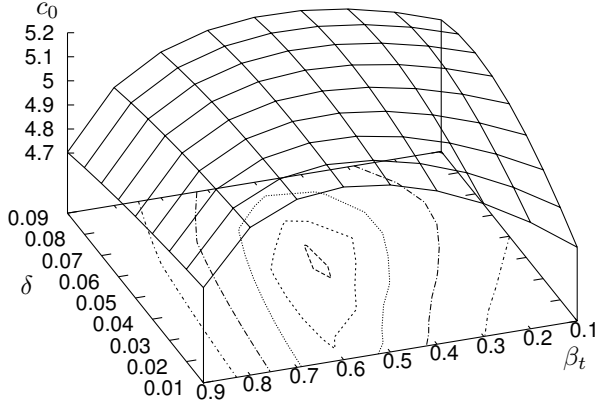}
\caption{A likelihood surface for the model $\theta_1$ with a contour plot
beneath.}
\label{fig:likelihood}
\end{center}
\end{figure}

Test model two
has an  outer model which connects a new 
node to either one or two distinct inner nodes
(equal probability of each).  The
inner model $\theta_2$ is given by
$\theta_2 = 0.25 \theta_0 + 0.25 \theta_t + 0.25 \theta_S +
0.25 \theta_D$.  Again 10,000 edges were generated using
this model.
A few test models with similar parameters to
$\theta_2$ are tested against $\theta_2$. 

\smallskip
\begin{center}
\begin{tabular}{@{\extracolsep{-10pt}}r rrrr r r | l} \hline
Model & & & &&&& $\quad c_0$\\ \hline
$\theta_2=0.25 \theta_0 $ & $+$ & $0.25 \theta_t $ & $+$ & $0.25 \theta_S $ & $+$ &
$0.25 \theta_D$ & \quad 2.45188 \\ \hline
$0.2 \theta_0 $ & $+$ & $0.3 \theta_t $ & $+$ & $0.25 \theta_S $ & $+$
& $0.25 \theta_D$ & \quad 2.43070 \\ 
$0.25 \theta_0 $ & $+$ & $0.25 \theta_t $ & $+$ & $0.3 \theta_S $ & $+$ &
$0.2 \theta_D$ &  \quad 2.43474 \\ 
$0.2 \theta_0 $ & $+$ & $0.25 \theta_t $ & $+$ & $0.3 \theta_S $ & $+$ &
$0.25 \theta_D$ & \quad 2.43549 \\ 
$0.24 \theta_0 $ & $+$ & $0.25 \theta_t $ & $+$ & $0.26 \theta_S $ & $+$ &
$0.25 \theta_D$ & \quad 2.45135 \\ \hline
\end{tabular}
\end{center}

As can be seen, even the final model which has extremely
close parameters produces a slightly lower $c_0$ value. With 
three free parameters in the model, an exhaustive 
state space search could quite time consuming.  If
the network were bigger, or more parameters were 
required in a test (a real network would not have
known model components), a brute-force state space 
search would be intractable.
For models with many parameters
the $c_0$ parameter could be used as a fitness function
for an optimisation procedure such as genetic algorithms.
Alternatively, for linear parameters, the GLM fitting 
from section \ref{sec:linear}
can be used and these tests are performed in the next
section.

\subsection{Testing the parameter optimisation}
\label{sec:GLMtest}

The next stage is to test the GLM fitting procedure described in
section \ref{sec:linear} on artificial models.  This can, in theory,
retrieve parameters from models produced by linear combinations
of model components.  In this section,
statistical significances from the GLM procedure
are quoted at the 10\%, 5\%, 1\%
or 0.1\% levels. 

First tests were performed on 
$\theta_1 = 0.5 \theta_p(0.05) + 0.5 \theta_t$
as described in the previous section.  The test
network again had 10,000 edges.  Sampling 
was used to generate just over 4,000,000 items of data
for the GLM fit.
Fitting $\theta= \beta_p \theta_p(0.05) + \beta_t \theta_t$ gave
the following results.

\smallskip
\begin{center}
\begin{tabular}{l | l l} \hline
Parameter & Estimate  & Significance \\ \hline
$\theta_p(0.05)$ & $0.53 \pm 0.031$ & $0.1\%$ \\
$\theta_t$ & $0.47 \pm 0.031$ & $0.1\%$ \\ \hline
\end{tabular}
\end{center}

The parameters were recovered almost exactly.
However, this assumed
that $\delta$ was known precisely.  If $\delta$ is not known then
the GLM procedure behaves reasonably with incorrect $\delta$.
The table below shows fits of the model with $\delta=0.2$ and
$\delta = 0.01$ -- considerably above and below the correct values.

\smallskip
\begin{center}
\begin{tabular}{l | l l} \hline
Parameter & Estimate  & Significance \\ \hline
$\theta_p(0.2)$ & $0.12 \pm 0.022$ & $0.1\%$ \\
$\theta_t$ & $0.84 \pm 0.021$ & $0.1\%$ \\ \hline
$\theta_p(0.01)$ & $0.43 \pm 0.025$ & $0.1\%$ \\
$\theta_t$ & $0.57 \pm 0.025$ & $0.1\%$ \\ \hline
\end{tabular}
\end{center}

In both cases the model correctly gave statistical significance
to the $\theta_p$ component of the model.  The actual estimates
were not 0.5, nor were them expected to be.  The
true $\delta$ parameter could be found by trying a range of
values within the GLM procedure just as it was with the likelihood
estimator in Figure \ref{fig:likelihood}.  

For realistic scenarios, the true underlying model is not known.  
Thus some
``misspecified" models (models known to be incorrect) were
tried to see whether incorrect components could be identified.  
Thus, the model
$\theta= \beta_d\theta_d + \beta_t \theta_t + \beta_0 \theta_0$ 
which includes extraneous $\theta_d$ (preferential attachment)
and $\theta_0$ (null or random) models.

\smallskip
\begin{center}
\begin{tabular}{l | l l} \hline
Parameter & Estimate  & Significance \\ \hline
$\beta_d$ & $0.46 \pm 0.057$ & $0.1\%$ \\
$\beta_t$ & $0.57 \pm 0.031$ & $0.1\%$ \\ 
$\beta_0$ & $-0.031 \pm 0.032$ & none \\  \hline
\end{tabular}
\end{center}

The $\theta_0$ component has been rejected having both a low value and a
low statistical significance.  The $\theta_d$ model has stayed in,
almost certainly because it has such a strong correspondence with
the $\theta_p(\delta)$ model -- indeed, for $\delta=0$ it is the
same model.

The GLM fitting procedure does not always produce the correct answer,
in particular, when $\theta_d$ and $\theta_p$ are included in
the same fitting procedure problems can occur.  
Fitting $\theta = \theta_d + \theta_p(0.05) + \theta_t$ gives
the following.

\smallskip
\begin{center}
\begin{tabular}{l | l l} \hline
Parameter & Estimate  & Significance \\ \hline
$\beta_d$ & $0.28 \pm 0.085$ & $0.1\%$ \\
$\beta_p(0.05)$ & $0.18 \pm 0.11$ & none \\ 
$\beta_t$ & $0.54 \pm 0.038$ & $0.1\%$ \\  \hline
\end{tabular}
\end{center}

Here the GLM procedure gave an incorrect answer.  The
$\theta_p(\delta)$ model was incorrectly rejected and given no
statistical significance.  This kind of error is common when
$\theta_d$ and $\theta_p(\delta)$ are combined in the same model.
This model gives $c_0 = 5.17$ compared with $c_0 = 5.18$ for the
correct model -- the likelihood still identifies the correct
model even when the GLM procedure fits an incorrect model.

The GLM procedure was next used to recover parameters from
$\theta_2 = 0.25 \theta_0 + 0.25 \theta_t + 0.25 \theta_S +
0.25 \theta_D$.  The test network had 10,000 edges as previously.
Sampling was used to obtain 
just over 3.5 million data points for model fitting.

\smallskip
\begin{center}
\begin{tabular}{l | l l} \hline
Parameter & Estimate  & Significance \\ \hline
$\beta_0$ & $0.23 \pm 0.021$ & $0.1\%$ \\
$\beta_t$ & $0.28 \pm 0.017$ & $0.1\%$ \\ 
$\beta_S$ & $0.24 \pm 0.016$ & $0.1\%$ \\
$\beta_D$ & $0.25 \pm 0.020$ & $0.1\%$  \\  \hline
\end{tabular}
\end{center}

As can be seen, this recovery of parameters was
quite successful, although $\beta_t$
is actually 0.25 and therefore slightly outside
the error range $0.28 \pm 0.017$.  The
next test was to add a spurious model component
$\theta_d$.

\smallskip
\begin{center}
\begin{tabular}{l | l l} \hline
Parameter & Estimate  & Significance \\ \hline
$\beta_0$ & $0.33 \pm 0.059$ & $0.1\%$ \\
$\beta_t$ & $0.29 \pm 0.017$ & $0.1\%$ \\ 
$\beta_S$ & $0.24 \pm 0.016$ & $0.1\%$ \\
$\beta_D$ & $0.23 \pm 0.022$ & $0.1\%$ \\ 
$\beta_d$ & $-0.089 \pm 0.059$ & $5\%$\\  \hline
\end{tabular}
\end{center}

The $\beta_d$ parameter was given a negative
value (which is likely to produce an invalid 
model for the likelihood estimate) 
and the relatively low 
statistical significance also suggests
$\theta_d$ 
should be removed from the model.  An important caveat
exemplified here is that the GLM model is not 
constrained to
produce the $\beta$ parameters in the range $(0,1)$. 
This needs to be considered when analysing model fitting.



In most
circumstances tested, the GLM model performed extremely
well.  When the correct model was tested, the correct results
were obtained and spurious model components were only accepted
if they correlated strongly with genuine model components.
The GLM model is a very useful tool for exploratory data analysis
but the likelihood framework remains the true test of model fit
to data.

\subsection{Tests on real data}
\label{sec:realdata}

Tests  on five different data sets are reported
in \cite{feta}.  Here, for
space reasons, only one network is reported, the RouteViews AS network,
a view of the AS topology collected by the
University of Oregon RouteViews 
project\footnote{\url{http://www.routeviews.org}}.  
The data set gives the growth
of the AS topology from 42,000 edges to over 90,000.  Throughout this section,
it is important to keep in mind the aim of this paper, to test the FETA framework.
The models described here are not claimed to be the best known models
for the network in question.  The PFP model \cite{zhou2004} with its special outer model gets
a closer match to the final network statistics.  The ORBIS model \cite{ORBIS} 
does not model evolution but is very good at matching statistics on a target network.
The model
presented here as ``best" is the best model found using the FETA framework with a simple
outer model.  
The claim being verified in this section is not that this is the best possible model
of the real network but 
that models can be assessed and optimised using the FETA framework without looking at any target statistics other than likelihood. 

Three inner models were compared to the RouteViews AS network.
The outer model was simple -- the choice of operation (add new node, add
link to new node or add inner edge) was exactly that sequence observed in the
real data.
The inner model $\theta_0$ was used 
as a base for comparison.  The other two models were 
a ``pure" PFP model (but without the PFP special outer
model) $\theta_p(0.005)$ and the ``best" 
model found which was $0.81\theta_p(0.014) + 0.17\theta_R(1)$ (PFP + ``recent") 
to connect new nodes
and $0.71\theta_d + 0.22\theta_R(1) + 0.07\theta_S$ (preferential attachment
+ ``recent" + singleton) to connect edges between
existing nodes.  
The PFP model $\theta_p(0.005)$
had $c_0 = 4.81$ and the
``best" model had $c_0 = 8.06$.  
From these results PFP and ``best" should
be a significant improvement on random and ``best" should
be better than PFP.
These modelling results should not be taken as a criticism of PFP as
described in \cite{zhou2004} since the special ``interactive
growth" outer model of that paper was not used (the focus here
is on the inner model).

\begin{figure*}[ht!]
\begin{center}
\includegraphics[width=5.5cm]{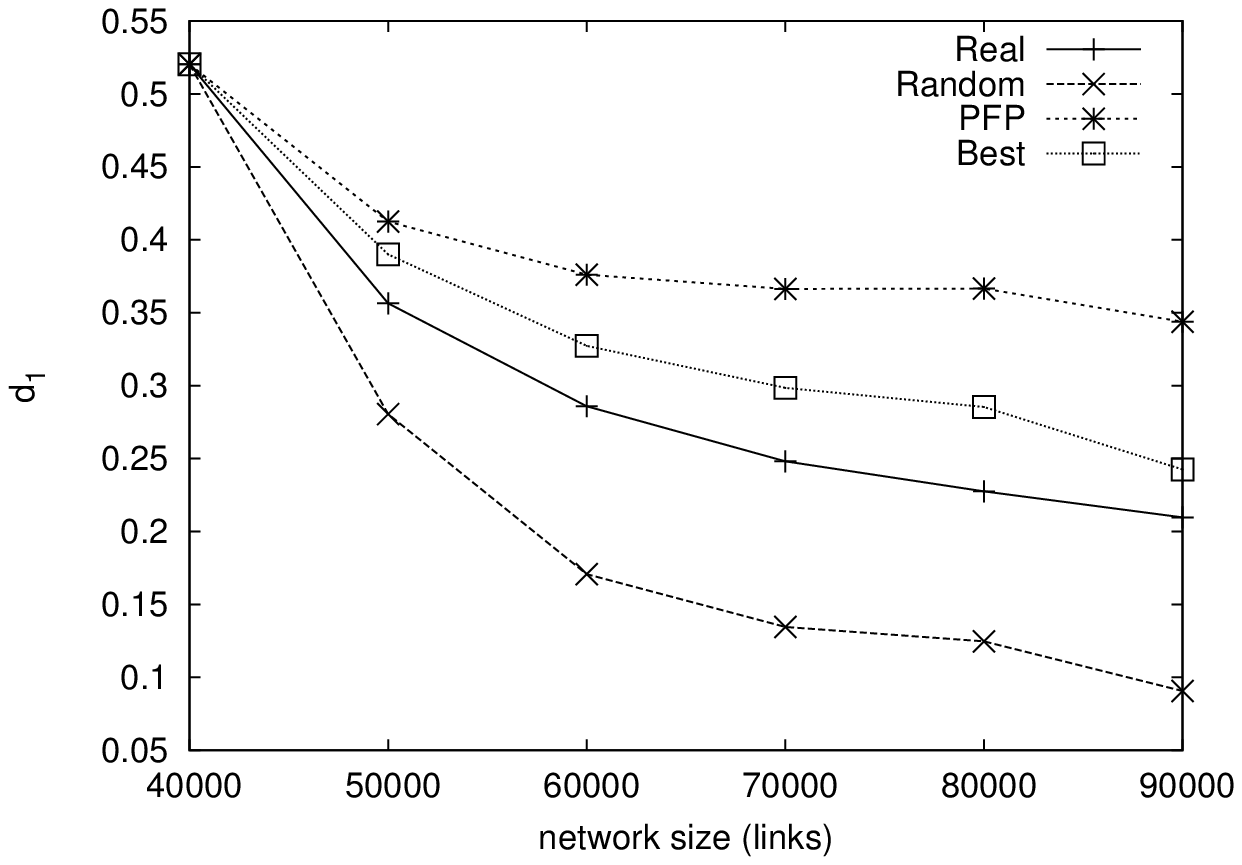}
\includegraphics[width=5.5cm]{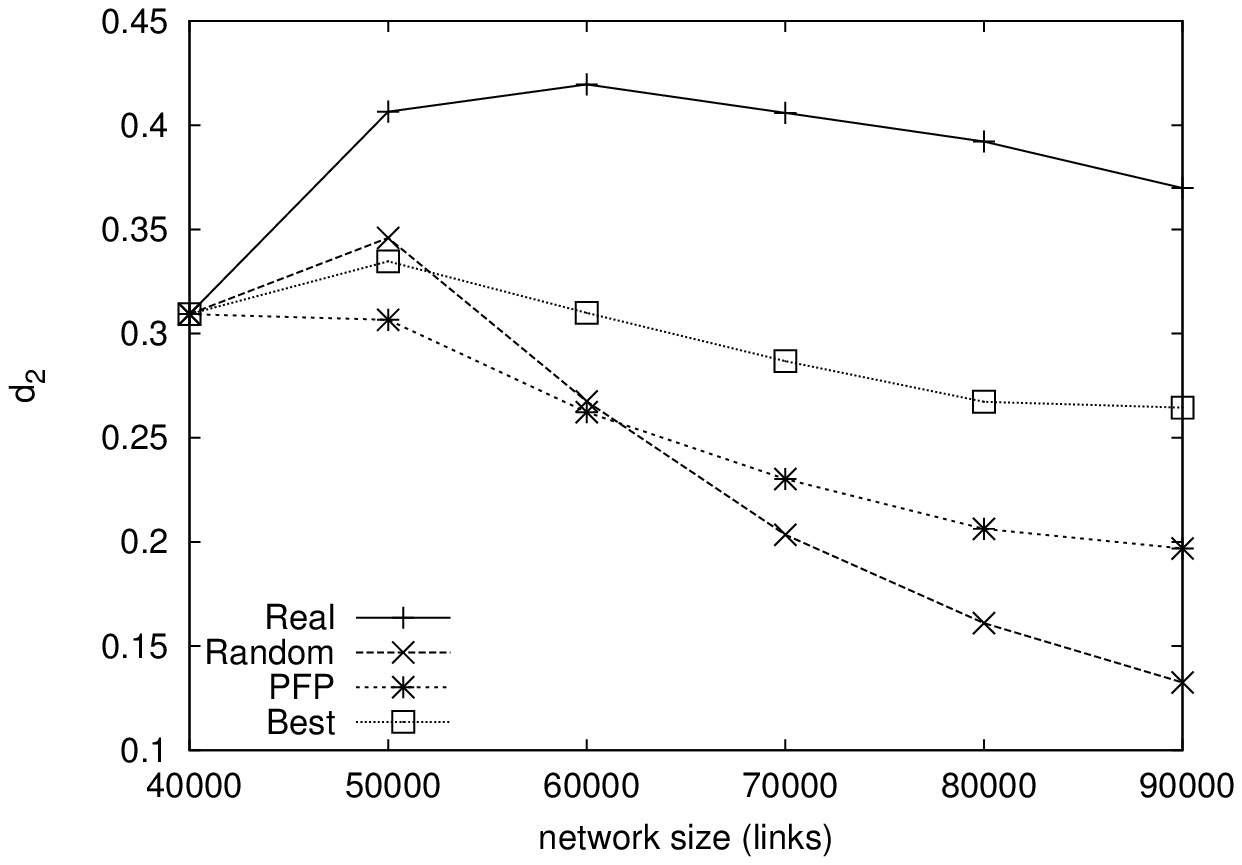}
\includegraphics[width=5.5cm]{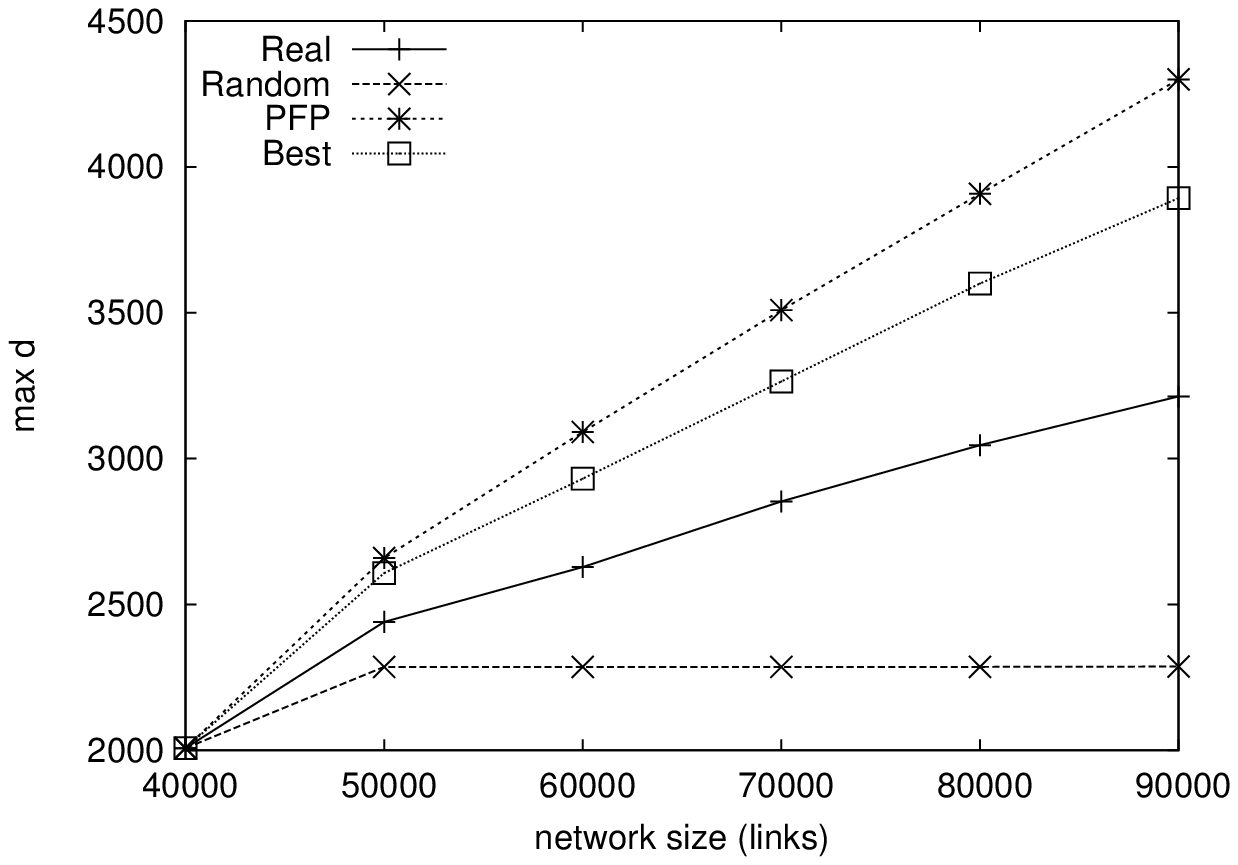}
\caption{The evolution of the $d_1$ (left), $d_2$ (center) and $\max d$ parameters.}
\label{fig:plot_d1}
\end{center}
\end{figure*}

\begin{figure*}[htb!]
\begin{center}
\includegraphics[width=5.5cm]{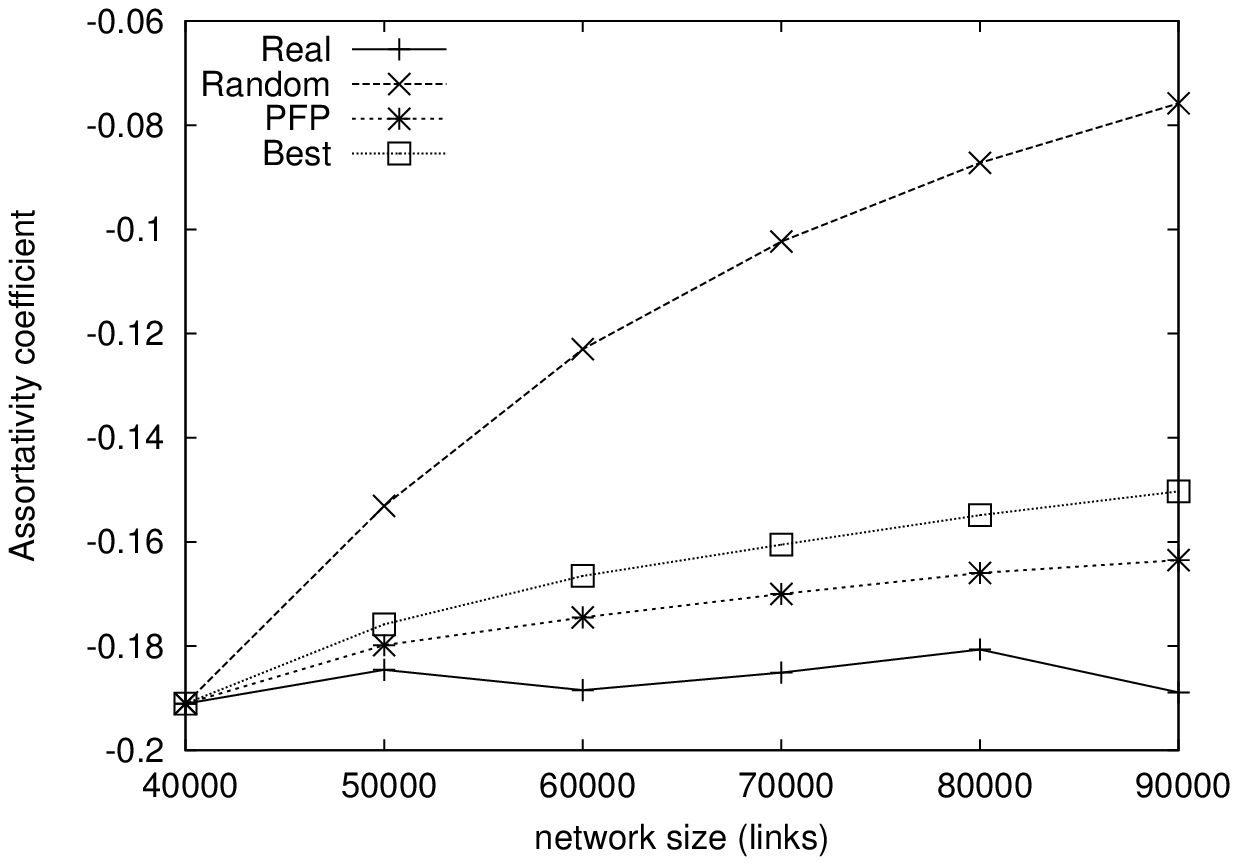}
\includegraphics[width=5.5cm]{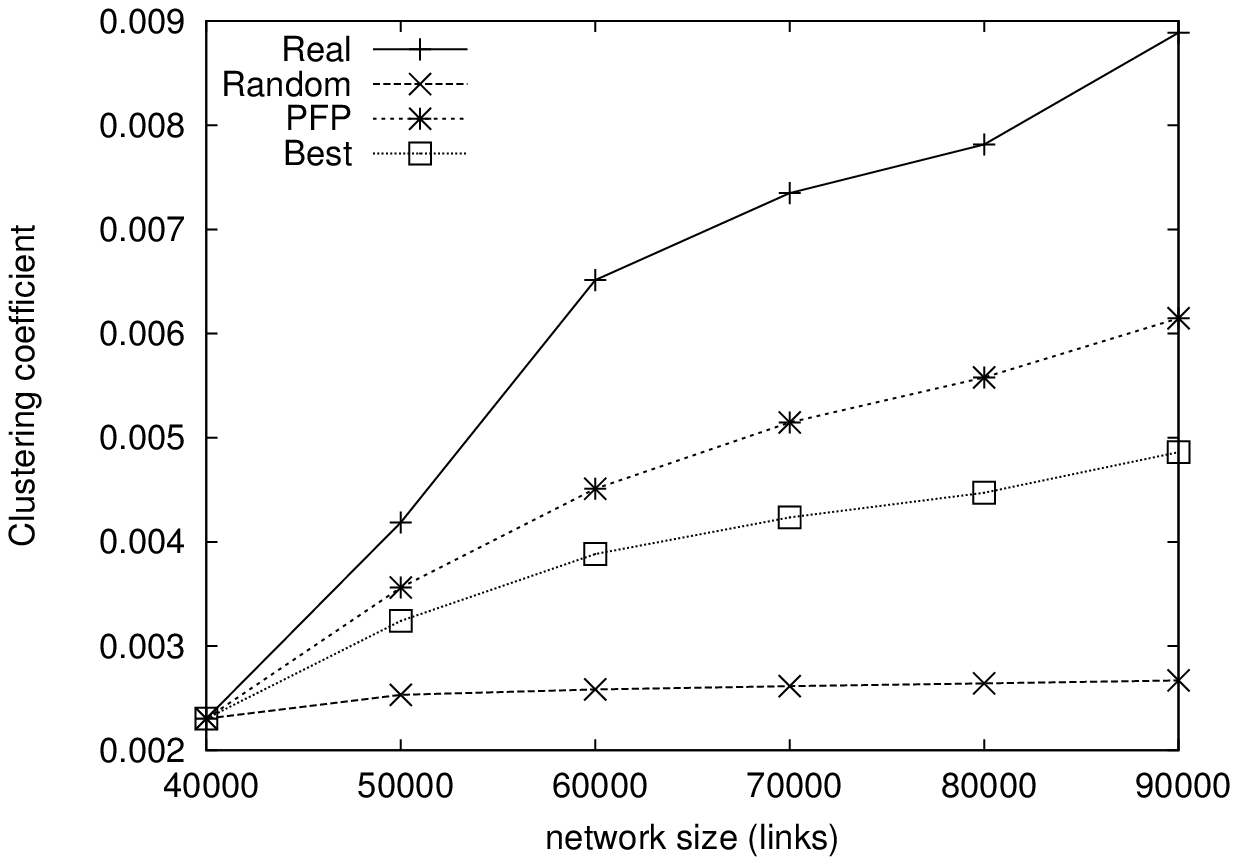}
\includegraphics[width=5.5cm]{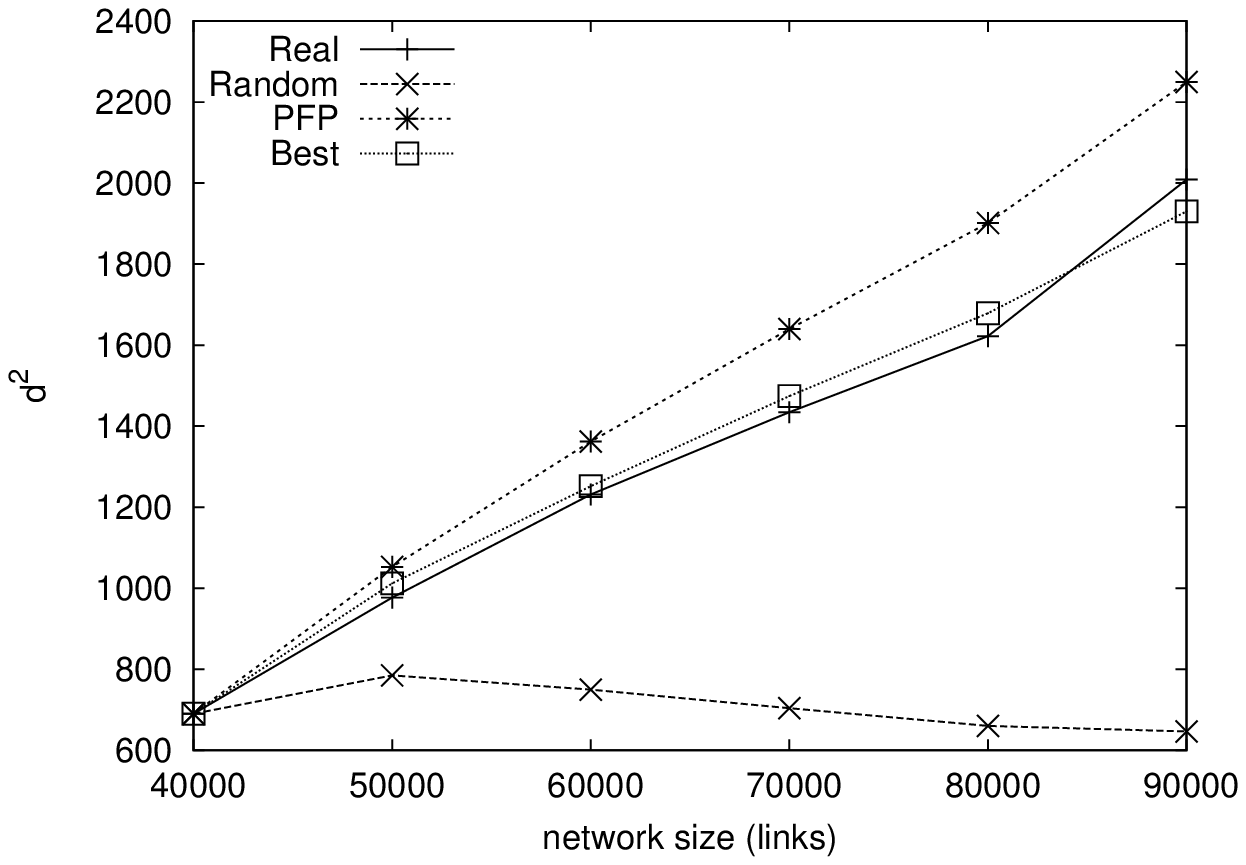}
\caption{The evolution of assortativity coefficient (left), clustering
coefficient (center) and $\overline{d^2}$ (right).}
\label{fig:plot_r}
\end{center}
\end{figure*}

Each model grew a test network from the seed network of 42,000 edges.
The first point in each plot is after edge 40,000 and hence shows
all models to perform the same (since the network is still the seed
network at this point).
Figures \ref{fig:plot_d1} and \ref{fig:plot_r} show the evolution of various
graph statistics for the real network compared with the three models.  
The leftmost point for each is within the seed graph and hence should
always be the same.  The statistics are $d_1$ and $d_2$ the proportion
of nodes of degree one and two, $\max d$ the degree of the highest node,
$\overline{d^2}$ (the mean square node degree), the assortativity coefficient
$r$ and the clustering coefficient $\gamma$.  See \cite{hamedsurvey} for
full descriptions of these statistics.  (Note that $\overline{d}$ is
fixed by the outer model and is an exact match to the real topology).

As mentioned at the start of this section, the claim is not that these
models are a perfect fit to the evolution of the target network but,
instead, that the order in which they fit the target network is that
given by the likelihood estimator:  the
``best" model being better than pure PFP, and both being much better
than random.  The models and the
$c_0$ measures which predicted this were produced
before any artificial topologies were generated and
without reference to the graph statistics plotted in the figures.  
This is a convincing demonstration
that the likelihood measure translates directly into fit to real data
over a number of statistical measures.

For most statistics, the ordering seems correct with ``best"
being closest to real, followed by PFP and then random.  An exception is in
the graphs for $\gamma$ and $r$
where PFP is slightly better than ``best".  However, in $d_1$ and
$\max d$ the PFP model is approximately the same as random,
when we would expect it to be better.  In the case of $\max d$, random predicts unrealistically 
slow growth.  For some statistics, no models given are
close (for reproducing the statistics of a graph snapshot it seems likely that ORBIS, for example, 
might be better).  However, the framework has clearly shown its ability
to assess which model best fits a target graph and this is clearly reflected in these
statistics.

\section{Conclusions}
\label{sec:conc}
The Framework for Evolving Topology Analysis
(FETA) is a useful toolset for investigating growth models 
of networks where evolution information is 
available.  Network growth models were described in terms
of an outer model (which selected the operation to perform on
the graph) and an inner model (which selected the entity
for the operation).  A likelihood statistic was given for
an inner model giving rise to a target network.
The likelihood statistic given is a rigorous and quick to
calculate.  It has been shown to recover the statistics
of a known model from a network grown using that model. 
A method was given for exploring and optimising linear combinations
of model components and this was tested successfully.
The fitting procedure can give insight into what model components
are required to best fit the data.  Models output by the fitting 
procedure can then be assessed precisely using the likelihood
measure.
FETA has been tested on
real data from five networks, one of which
was presented in this paper.  The likelihood
measure was found to be a good predictor of how well a
network grown from a given model would match the
statistics of the real data.  
The models presented here were not perfect
at capturing the evolution of the AS graph.  Different
inner model components would 
be needed to improve this.

Much more
can be achieved with the statistical analysis of network growth.  
A similar likelihood approach could be
applied to the outer model. Inner models which themselves change in time 
would be another improvement.  Models constructed
multiplicatively from components 
($\theta_1^{\beta_1}\theta_2^{\beta_2}\cdots$)
would seem natural than but 
normalisation problems exist.  Network models could be considered which
remove nodes or edges as well as add them and which do not necessarily
remain connected.  Finding new data sets to apply the method
to is also a priority.  Other researchers are encouraged to download
and try the software and 
data\footnote{\url{http://www.richardclegg.org/software/FETA}}.

\bibliographystyle{IEEEtran}
\bibliography{netsci_nsrl_2009}

\end{document}